\newif\iflong
\newif\ifshort

\longtrue

\iflong
\else
\shorttrue
\fi

\documentclass[a4paper]{article}
\usepackage{amsmath,amssymb, amsthm}

\usepackage[disable]{todonotes}
\usepackage{hyperref}
\usepackage{cleveref}
\usepackage{algorithmicx}
\usepackage[ruled]{algorithm}
\usepackage[noend]{algpseudocode}
\usepackage{xspace}
\usepackage{xparse}
\usepackage{semantic}
\usepackage{enumerate}
\usepackage{enumitem}
\usepackage{calc}
\usepackage{tikz}
\usepackage{pgfplots}
\usepackage{mathtools}
\usepackage{framed}
\usepackage{paralist}
\usepackage{caption}
\usepackage{wrapfig}
\usepackage{tabularx}
\usepackage[subrefformat=parens]{subcaption}
\usetikzlibrary{decorations,calc,shapes,snakes}

\usepackage{xcolor}

\hypersetup{
	colorlinks=true,
	linkcolor=blue,
	citecolor=blue}

\newif\ifcomment
\commentfalse
\commenttrue

\theoremstyle{plain}
\newtheorem{lemma}{Lemma}
\newtheorem{theorem}{Theorem}
\newtheorem{corollary}{Corollary}
\theoremstyle{definition}
\newtheorem{definition}{Definition}
\colorlet{commentcolour}{green!50!black}
\definecolor{magenta}{rgb}{0.79,0.08,0.48}
\definecolor{darkgreen}{rgb}{0.0, 0.5, 0.0}
\definecolor{darkblue}{rgb}{0.0, 0.06, 0.54}
\definecolor{pistachio}{rgb}{0.58, 0.77, 0.45}
\definecolor{cornflowerblue}{rgb}{0.39, 0.58, 0.93}
\definecolor{dandelion}{rgb}{0.94, 0.88, 0.19}

\newcommand{\variableStyle}[1]{{\mathsf{#1}}}

\newcommand{\Set}[1]{\left\lbrace #1 \right\rbrace}
\newcommand{\quantify}{:}
\newcommand{\NP}{\variableStyle{NP}}
\newcommand{\DTIME}{\variableStyle{DTIME}}
\newcommand{\coNP}{\variableStyle{coNP}}

\newcommand{\WClass}[1]{\variableStyle{W}[#1]}
\newcommand{\FPT}{\variableStyle{FPT}}

\newcommand{\coNPPoly}{\coNP/\variableStyle{poly}}
\newcommand{\trees}{\variableStyle{trees}}
\newcommand{\VC}{\variableStyle{VC}}

\newcommand{\N}{\mathds{N}}
\newcommand{\Solution}{\mathfrak{P}}
\newcommand{\Instance}[1]{\mathcal{#1}}
\newcommand{\Probl}[1]{\textsc{#1}}

\newcommand{\MSElong}{\Probl{Minimum Shared Edges}}
\newcommand{\DMSElong}{\Probl{Directed Minimum Shared Edges}}

\newcommand{\diamG}[1]{\ifx\relax #1\relax
	\variableStyle{diam}
\else
	\variableStyle{diam}_{#1}
\fi} %
\newcommand{\distG}[3][G]{\variableStyle{dist}_{#1}(#2,#3)} %
\renewcommand{\deg}[1]{\variableStyle{deg}(#1)}

\newcommand{\Abs}[1]{\left| #1 \right|}
\newcommand{\Floor}[1]{\left\lfloor #1 \right\rfloor}
\newcommand{\Ceil}[1]{\left\lceil #1 \right\rceil}
\newcommand{\Probability}[2]{\ifx\relax#2\relax
	\operatorname{Pr}^{#1}
\else
	\operatorname{Pr}^{#1}\left[#2\right]
\fi}

\newcommand{\inlinefrac}[2]{\frac{#1}{#2}}

\newcommand{\defDecprob}[3]{
	\begin{center}
		\begin{minipage}{0.95\textwidth}
			\noindent
			\textsc{#1}\\
			\setlength{\tabcolsep}{3pt}
			\begin{tabularx}{\textwidth}{@{}lX@{}}
					\normalsize \textbf{Input:} 		& \normalsize #2 \\
					\normalsize \textbf{Question:} 	& \normalsize #3
				\end{tabularx}
		\end{minipage}
	\end{center}
}

\newcommand{\grid}[2]{G_{#1\times#2}}
\newcommand{\rimDist}[2]{\partial_{#1} #2}
\newcommand{\rimDistDual}[2]{\overline{\partial}_{#1} #2}

\newcommand{\Path}{\operatorname{P}}
\newcommand{\rainbow}[2]{
		\node (start) [circle, scale=0.3,fill] at #1 {};
		\path (start) edge ($(start) + (3*#2/8,0)$);
		\node (end) [circle,scale=0.3,fill] at ($(start)+(#2,0)$) {};
		\path ($(end) - (3*#2/8,0)$) edge (end);
		\foreach\i in {1,...,3}
		{
			\path ($(start)+(\i*#2/8,0)$) edge [out=90,in=90,bend angle=180,looseness=1.7] ($(end)+(-\i*#2/8,0)$);
		}
	}
\newcommand{\VCgrid}[4]{
\begin{scope}[shift={#1},xscale={#2},yscale={#3}]

		\ifthenelse{#4=1 \or #4=2 \or #4=3}
		{\node [circle, scale=0.6,fill,green] at (0,3) {};}{}	
		\ifthenelse{#4=2}
		{\path [draw,green,line width=3pt] (0,3) -- (1,3);}{}
		\ifthenelse{#4=2}
		{\path [draw,green,line width=3pt] (3,3) -- (4,3);}{}
		\ifthenelse{#4=3}
		{\path [draw,red,line width=3pt] (1,3) -- (2,3);}{}
		\ifthenelse{#4=3}
		{\path [draw,red,line width=3pt] (2,3) -- (3,3);}{}
		
		\ifthenelse{#4=4}
		{
			\draw[snake=snake,green, line width=3pt] (0,2) -- (0,3);
			\path [draw,green,line width=3pt] (0,2) -- (1,2);
			\path [draw,green,line width=3pt] (1,2) -- (2,2);
			\draw[snake=snake,green, line width=3pt] (2,1) -- (2,2);
			\path [draw,green,line width=3pt] (2,1) -- (3,1);
			\path [draw,green,line width=3pt] (3,1) -- (4,1);
			\draw[snake=snake,green, line width=3pt] (4,0) -- (4,1);
			
		}{}
		
		\foreach\i in {0,...,4} 
			\foreach\j in {0,...,3}\node [circle, scale=0.3,fill] at (\i,\j) {};
		\foreach\i in {1,...,4} \node at (-0.3,3.83-\i) {$v'_{\i,1}$};
		\foreach\i in {1,...,4} \node at (\i-0.5,-0.3) {$e_{\i}$};
		\path [draw] (0,3) -- (1,3);
		\path [draw] (1,3) -- (1.2,3);
		\rainbow{(1.2,3)}{0.8};
		\path [draw] (2,3) -- (2.2,3);
		\rainbow{(2.2,3)}{0.8};
		\path [draw] (3,3) -- (4,3);
		\path [draw] (0,2) -- (1,2);
		\path [draw] (1,2) -- (2,2);
		\path [draw] (2,2) -- (2.2,2);
		\rainbow{(2.2,2)}{0.8};
		\path [draw] (3,2) -- (3.2,2);
		\rainbow{(3.2,2)}{0.8};
		\path [draw] (0,1) -- (0.2,1);
		\rainbow{(0.2,1)}{0.8};
		\path [draw] (1,1) -- (2,1);
		\path [draw] (2,1) -- (3,1);
		\path [draw] (3,1) -- (4,1);
		\path [draw] (0,0) -- (0.2,0);
		\rainbow{(0.2,0)}{0.8};
		\path [draw] (1,0) -- (1.2,0);
		\rainbow{(1.2,0)}{0.8};
		\path [draw] (2,0) -- (3,0);
		\path [draw] (3,0) -- (3.2,0);
		\rainbow{(3.2,0)}{0.8};
		\draw[snake=snake] (0,0) -- (0,1);
		\draw[snake=snake] (0,1) -- (0,2);
		\draw[snake=snake] (0,2) -- (0,3);
		\draw[snake=snake] (1,0) -- (1,1);
		\draw[snake=snake] (1,1) -- (1,2);
		\draw[snake=snake] (1,2) -- (1,3);
		\draw[snake=snake] (2,0) -- (2,1);
		\draw[snake=snake] (2,1) -- (2,2);
		\draw[snake=snake] (2,2) -- (2,3);
		\draw[snake=snake] (3,0) -- (3,1);
		\draw[snake=snake] (3,1) -- (3,2);
		\draw[snake=snake] (3,2) -- (3,3);
		\draw[snake=snake] (4,0) -- (4,1);
		\draw[snake=snake] (4,1) -- (4,2);
		\draw[snake=snake] (4,2) -- (4,3);
	\end{scope}
}
\newcommand{\gridPath}[2]{
			\gettikzxy{#2}{\goalX}{\goalY}
			\gettikzxy{#1}{\startX}{\startY}
			\pgfmathsetmacro{\middle}{\goalX-\startX}
			\path ($(\startX,\startY)$) edge ($(\startX+\middle/2,\startY)$);
			\path ($(\startX+\middle/2,\startY)$) edge ($(\startX+\middle/2,\goalY)$);
			\path ($(\startX+\middle/2,\goalY)$) edge ($(\goalX,\goalY)$);
		}
		
\makeatletter
\newcommand{\gettikzxy}[3]{%
  \tikz@scan@one@point\pgfutil@firstofone#1\relax
  \edef#2{\the\pgf@x}%
  \edef#3{\the\pgf@y}%
}
\makeatother		

\newcommand{\tw}{\variableStyle{tw}}
\newcommand{\polyEquivRel}{\mathcal{R}}
\newcommand{\OrCross}{\textsc{OR}-cross-composition}

\renewcommand{\N}{\mathbb{N}}

\newcommand{\DMSE}{{DMSE}}
\newcommand{\MSE}{{MSE}}

\newcommand{\appsymb}{$\star$}
\newcommand{\appref}[1]{{\hyperref[proof:#1]{\appsymb}}}

\newcommand{\appendixproof}[3]{%
  \iflong{}{#3}\else\gappto{\appendixProofText}{\subsection{Proof of #1~\ref{#2}}\label{proof:#2}#3}\fi{}
}

\newcommand{\appendixsection}[1]{%
  \iflong{}{}\else\gappto{\appendixProofText}{\section{Additional Material for Section~#1}}\fi{}
}

\newcommand{\raproof}{($\Rightarrow$)}
\newcommand{\laproof}{($\Leftarrow$)}

\newcommand{\sharedEdges}[2]%
{\ifthenelse{\equal{#2}{up}}
    {k_{#1}^{\uparrow}}
    {\ifthenelse{\equal{#2}{down}}
	    {k_{#1}^{\downarrow}}
	    {\ifthenelse{\equal{#2}{right}}
		    {k_{#1}^{\rightarrow}}
		    {\ifthenelse{\equal{#2}{left}}
			    {k_{#1}^{\leftarrow}}
			    {\textcolor{red}{which $k$?}}}}}%
}

\def\dotminus{\mathbin{\ooalign{\hss\raise1ex\hbox{.}\hss\cr
  \mathsurround=0pt$-$}}}

\usepackage{authblk}
\usepackage{etoolbox}

\title{The Minimum Shared Edges Problem on Grid-like Graphs}
\author[1]{Till~Fluschnik\thanks{Supported by the DFG, project DAMM (NI~369/13-2).}} 
\author[1]{Meike~Hatzel}
\author[1]{Steffen~H\"artlein}
\author[1]{Hendrik~Molter\thanks{Partially supported by the DFG, project DAPA (NI~369/12).}}
\author[1]{Henning~Seidler}
\affil[1]{\small{Institut f\"ur Softwaretechnik und Theoretische Informatik, TU~Berlin, Germany, \texttt{\{till.fluschnik,meike.hatzel,hendrik.molter\}@tu-berlin.de}\\ \texttt{\{haertlein,henning.seidler\}@campus.tu-berlin.de}}}
\date{\vspace{-5ex}}

\begin{document}

\mathligsoff %

\maketitle
\begin{abstract}
	\noindent We study the $\NP$-hard \emph{Minimum Shared Edges} (MSE) problem on graphs:
	decide whether it is possible to route~$p$~paths from a start vertex to a target vertex in a given graph while using at most~$k$~edges more than once.
	We show that~MSE can be decided on bounded (i.e.\ finite) grids in linear time when both dimensions are either small or large compared to the number~$p$ of paths.
	On the contrary, we show that~MSE remains $\NP$-hard on subgraphs of bounded grids.
	
	Finally, we study \MSE{} from a parametrised complexity point of view. 
	It is known that \MSE{} is fixed-parameter tractable with respect to the number~$p$ of paths. 
	We show that, under standard complexity-theoretical assumptions, the problem parametrised by the combined parameter~$k$,~$p$, maximum degree, diameter, and treewidth does not admit a polynomial-size problem kernel, even when restricted to planar graphs.
\end{abstract}

\section{Introduction}

Routing in street-like networks is a frequent task.
Graphs modelling street networks are often (almost) planar, that is, they can be drawn in the plane with (almost) no edge crossings.
As a special case, a graph modelling the street network in Manhattan is similar to a grid graph.
We study the following problem, originally introduced by Omran et al.~\cite{omran2013MSE}, from a computational (parametrised) complexity perspective on planar and grid-like graphs:

\defDecprob{\MSElong{} (\MSE{})} 
{An undirected graph $G=(V,E)$, two distinct vertices $s,t\in V$, and two integers~$k, p \in \N$.}
{Are there $p$ paths from $s$ to $t$ in~$G$ such that at most~$k$ edges appear in more than one of the $p$ paths?}
Note that Omran et al.~\cite{omran2013MSE} originally defined the problem on directed graphs (we refer to this as \DMSElong{} or \DMSE{}). 
While Omran et al. motivate \MSE{} by applications in security management,
the problem can further appear in the following scenario. 
A network company wants to upgrade their network since it still uses old copper cables.
To improve the throughput, some of these cables shall be replaced by modern optical fibre cables.
The network routes information from a source location to a target location and the company wants to achieve a certain minimal throughput.
Since digging up the conduits for the cables is much more expensive than the actual cables, we can neglect the cost of the cables and upgrade them to arbitrary bandwidth, because once open, we can lay as many cables as necessary into a conduit.
The company wants to find the minimum number of conduits that have to be dug up in order to achieve the desired bandwidth.

\paragraph{Related Work.}
Omran et al.~\cite{omran2013MSE} showed that \DMSE{} is $\NP$-complete on directed acyclic graphs. 
The problems \MSE{} and \DMSE{} were both shown to be $\NP$-complete even if the input graph is planar~\cite{fluschnik2016planar}.
Moreover, \MSE{} is solvable in linear time on unbounded (i.e.\ infinite) grid graphs~\cite{MT_Fluschnik_MSE_2015}.  
\DMSE{} is $\Floor{k/2}$-approximable~\cite{AssadiENYZ14}, but there is no
polynomial-time approximation of factor %
$2^{(\log(n))^{1-\epsilon}}$ for any $\epsilon > 0$ unless 
$\NP \subseteq \DTIME(n^{\operatorname{polylog}(n)})$~\cite{omran2013MSE}.

Analysing its parametrised complexity,
Fluschnik et al.~\cite{fluschnik2015para_MSE} showed that \MSE{} is fixed-parameter tractable when parametrised 
by the number $p$ of paths but does not admit a polynomial-size problem kernel unless 
$\NP \subseteq \coNPPoly$,  \MSE{} is $\WClass{1}$-hard when parametrised by $\tw+k$, where~$\tw$ denotes the treewidth of the input graph,
and $\WClass{2}$-hard when parametrised by the number~$k$ of shared edges.
Furthermore, \MSE{} is solvable in polynomial time on graphs of bounded treewidth~\cite{AokiHHIKZ16,ye2013boundedTW}. %

\paragraph{Our Contribution.}
We give both positive and negative results for \MSE{} on grid-like graphs.
On the positive side, we show that if the dimensions of the grid are smaller than the number~$p$ of paths, then \MSE{} is trivially decidable, and if the dimensions of the grid are at least the number~$p$ of paths, then we provide an arithmetic criterion to decide~\MSE{} in linear-time (Section~\ref{sec:bounded_grids}).
On the negative side, we prove that the situation changes when subgraphs of bounded grids (which we refer to as \emph{holey grids}) are considered, that is, we prove that \MSE{} on subgraphs of bounded grids is $\NP$-hard (Section~\ref{sec:holey}).
Similarly, we prove that \DMSE{} is $\NP$-hard for acyclic subgraphs of directed bounded grids (Section~\ref{sec:manhattan}).
Our $\NP$-hardness results improve upon the known hardness results~\cite{fluschnik2016planar} as the graphs we consider are more restricted.
Moreover, we show that \MSE{} parametrised by~$k+p+\Delta+ \diamG{} +\tw$, where~$\Delta$ and~$\diamG{}$ denote the maximum degree and diameter, respectively, does not admit a polynomial-size problem kernel, unless $\NP \subseteq \coNPPoly$, even on planar graphs (Section~\ref{sec:polyKernel}), improving an existing kernelization lower bound~\cite{fluschnik2015para_MSE}.

\section{Preliminaries}

We use basic notation from graph theory and parametrised complexity.
We denote by~$\N$ the natural numbers containing zero.

\paragraph{Graph Theory.}

Unless stated otherwise, we assume that all graphs are finite, undirected, simple and without self-loops.
We refer with $V(G)$ and $E(G)$ to the vertex set and edge set, respectively, of a graph $G$. An edge set $P \subseteq E$ is called a \emph{path} if we have $P = \Set{\Set{v_{i-1},v_i} \mid 0 < i\leq n}$ for some pairwise distinct vertices $v_0,\ldots,v_n$.
In this case we say $P$ is a $v_0$-$v_n$-path of length $n$.
The distance~$\distG{u}{v}$ between two vertices $u, v\in V(G)$ is defined as the length of a shortest $u$-$v$-path (we set~$\distG{u}{v}=\infty$ if there is no $u$-$v$~path in~$G$).

\paragraph{Grids.}
For $n,m \in \N$, let $\grid{n}{m}$ be the (bounded) $n \times m$-grid, that is, the 
undirected graph $(V,E)$ with the set of vertices $V \coloneqq \Set{(x,y) \in \N \times \N \mid x < n, y < m}$ 
and the set of edges ${E\coloneqq \{\Set{(v,w),(x,y)} \mid |v-x|+|w-y|=1\}}$.
The \emph{coordinates} of a vertex are denoted by $v\coloneqq(v_x,v_y)$.
We call the vertices of degree less than four the \emph{rim} of the grid.
\iflong{}We refer to the vertex set $\Set{(x,m-1) \mid 0 \leq x < n}$ as the 
\emph{upper} rim and analogously we use \emph{lower}, \emph{left} and \emph{right} rim.\fi{}
For a given vertex~$v \in V$ we define~$\rimDist{x}{v} \coloneqq v_x$ and 
$\rimDist{y}{v} \coloneqq v_y$, $\rimDistDual{x}{v} \coloneqq n-1 - v_x$, and 
$\rimDistDual{y}{v} \coloneqq m-1 - v_y$. 
We also use~$\rimDist{}{v} \coloneqq \rimDist{x}{v} + \rimDist{y}{v}$ and~
$\rimDistDual{}{v} \coloneqq \rimDistDual{x}{v} + \rimDistDual{y}{v}$.
\paragraph{Parametrised Complexity.}
A pair $Q \mathop{=} (P,\kappa)$ with $P \mathop{\subseteq} \Sigma^*$ and $\kappa\mathop{:} \Sigma^* \mathop{\to} \N$ is called a \emph{parametrised problem}.
A parametrised problem~$Q \mathop{=} (P,\kappa)$ admits a \emph{problem kernel} (or is \emph{kernelisable}) if there is a polynomial-time algorithm transforming any instance~$\Instance{I}$ of~$Q$ into an instance~$\Instance{I}'$ such that 
\begin{inparaenum}[(i)]
		\item $\Instance{I} \in Q \Leftrightarrow \Instance{I}' \in Q $, and
		\item the size of $\Instance{I}'$ (the \emph{kernel}) is bounded by a computable function $f(\kappa(\Instance{I}))$.
\end{inparaenum}
If $f$ is a polynomial, then the problem is said to admit a polynomial (problem) kernel.
A parametrised problem is \emph{fixed-parameter tractable} (or in FPT) if each instance $(x,\kappa)$ can be decided in~$f(\kappa(x))\cdot |x|^{O(1)}$~time, where $f$ is a computable function.
A (decidable) parametrised problem is in FPT if and only if it is kernelisable.
A parametrised problem that is~$\variableStyle{W}$-hard is presumably not in FPT. 

\paragraph{Further Notation.}
Let $\Instance{I} = (G,s,t,p,k)$ be an instance of \MSE.
We say $\Solution$ is a \emph{solution} for $\Instance{I}$ if 
\begin{inparaenum}[(i)]
	\item $\Solution$ is a multiset of $p$ $s$-$t$-paths~$\{P_1,\ldots,P_p\}$, and
	\item $|\{e\in E \mid  \exists 1\leq i<j\leq p \quantify e \in P_i\cap P_j\}| \leq k$.
\end{inparaenum}
We say that $\Solution$ is a \emph{trivial solution} if $P_i=P_j$ for all $i,j\in[p]$. 
An edge is called \emph{shared} if it occurs in at least two paths of~$\Solution$.

\section{On Bounded and Holey Grids}
\label{sec:grids}
\appendixsection{\ref{sec:grids}}

\noindent The class of grid graphs appeared frequently in the literature:
There is work on grid graphs and related graphs with respect to finding paths~\cite{Kanchanasut94,KanteMMN15}, routing~\cite{BhatiaLMN92}, or structural properties~\cite{AlankoCIOP11,Jelinek10}.
In this section we study the complexity of \MSE{} on bounded grids and their subgraphs.
We show that \MSE{} is solvable in linear time on bounded grids when both dimensions are either small or large compared to the number~$p$ of paths (Section~\ref{sec:bounded_grids}) and becomes $\NP$-hard for subgraphs of bounded grids (Sections~\ref{sec:holey} and~\ref{sec:manhattan}).
We remark that \MSE{} is solvable in linear-time on the class of unbounded grids~\cite{MT_Fluschnik_MSE_2015}.

\subsection{Bounded Grids}
\label{sec:bounded_grids}

We fix some instance $\Instance{I} \coloneqq (G = G_{n\times m},s,t,p,k)$ for the remainder of the section.
Since the problem is invariant under symmetry and swapping~$s$ and~$t$, we may assume~$s$ lies left and below of~$t$ and $\rimDist{x}{s}\leq\rimDist{y}{s}$.
To show optimality of the constructions we regard edge cuts of size less than $p$.
Assume $\Instance{I}$ has a solution~$\Solution$.
We know~\cite{MT_Fluschnik_MSE_2015} that after contraction of the shared edges, the graph must allow an $s$-$t$-flow of value at least~$p$.
Therefore, every cut smaller than $p$ has to be eliminated by a contraction, that is, it must contain a shared edge.

We distinguish the following different cases depending on the dimensions of the grid in relation to the number~$p$ of paths: \emph{$p$-small} grid ($p> \max\{n,m\}$), \emph{$p$-large} grid ($p\leq\min\{n,m\}$), and \emph{$p$-narrow} grids (neither $p$-small nor $p$-large).
We leave open whether \MSE{} is solvable in polynomial-time on $p$-narrow grids.
However, ongoing work indicates that the question can be answered positively.
\paragraph{On \boldmath$p$-small grids.}
If $p> \max\{n,m\}$, then every set of horizontal edges with endpoints having the same coordinates in the grid forms an $s$-$t$-cut of size smaller than~$p$ (analogously for every set of horizontal edges).
Hence, intuitively, any set of $p$ $s$-$t$-paths share an edge for each horizontal or vertical level they cross.
Indeed, we prove that every instance on $p$-small grids is a yes instance if and only if it admits the trivial solution.
\begin{lemma}%
	\label{lem:small_grid}
	If $m<p$ and $n < p$, then we have a solution if and only if $\distG{s}{t}\leq k$.
\end{lemma}
\appendixproof{Lemma}{lem:small_grid}
{
  \begin{proof}
	  Each row and each column between $s$ and $t$ is a cut of size smaller than $p$.
	  Formally, these are rows $\variableStyle{row}_x \coloneqq \{\{(x,y),(x+1,y)\}\mid 0\leq y < m\}$ for $s_x \leq x <t_x$ and columns $\variableStyle{col}_y \coloneqq \{\{(x,y),(x,y+1)\}\mid 0\leq x < n\}$ for $s_y \leq y <t_y$.
	  All these cuts are disjoint, since neither two rows nor two columns nor a row and a column ever intersect with respect to edges.
	  Their number is $(t_x-s_x)+(t_y-s_y)=\distG{s}{t}$.
	  So if we have a solution, then $k$ has to be at least $\distG{s}{t}$. 
	  Conversely, if $\distG{s}{t}\leq k$, we have a trivial solution by definition.
  \end{proof}
}

\paragraph{On \boldmath$p$-large grids.}

Compared to the situation on $p$-small grids, $p$-large grids allow for non-trivial solutions.
Nevertheless, we prove that the existence of such non-trivial solutions is expressed by arithmetic conditions which can be checked in linear time.
These arithmetic conditions basically relate $p$, $k$, and the positions of $s$ and $t$ relative to the rim of the grid.
If $s$ lies sufficiently far away from the corner formed by the left and lower rim, then only every second path in our construction introduces a new shared edge at this part.
However, if~$s$ lies close to the corner (or if~$p$ is large enough), there is a critical number of paths after which every additional path introduces at least one new shared edge.
The same happens at the side of $t$.
Thus we obtain the following cases.

\begin{lemma}%
	\label{lem:constr_large_k}
	Let $p\leq m$ and $p\leq n$.
	Then there is a non-trivial solution if and only if either
	\begin{compactitem}
		\item $p \leq 2 (\rimDist{}{s} + 2) - \deg{s}$ 
			and $k \geq \Ceil{\inlinefrac{1}{2}(p-\deg{s})} + \Ceil{\inlinefrac{1}{2}(p-\deg{t})}$, or
		\item $2 (\rimDist{}{s} + 2) - \deg{s} < p \leq 2 (\rimDistDual{}{t} + 2) - \deg{t}$\\
			and $k \geq p - (\rimDist{}{s}+2) +\Ceil{\frac{1}{2}(p-\deg{t})}$, or
		\item $p > 2 (\rimDistDual{}{t} + 2) - \deg{t}$ 
			and $k \geq 2p -(\rimDist{}{s} + \rimDistDual{}{t} + 4)$.
	\end{compactitem}
\end{lemma}

{

To simplify the notation for our construction, we introduce the following arrow notation.
For $(x,y) \in V$ we define $(x,y) \rightarrow (x+ \ell,y) \coloneqq \{\{(x+i,y),(x+i+1,y)\} \in E \mid 0 \leq i < \ell\}$.
Analogously we define $\uparrow, \downarrow$ and $\leftarrow$.
We also use the concatenation of these expressions such that e.g.~$u \rightarrow v \uparrow w \coloneqq (u \rightarrow v) \cup (v \uparrow w)$. %
\begin{figure}[t]
  \centering
  \begin{subfigure}[b]{.475\textwidth}
    \centering
    \begin{tikzpicture}[scale=0.35]
	    \path [use as bounding box] (-0.5,-0.5) rectangle (16.5,12.7);
	    \draw [help lines] grid (10.9,9.9);
	    \foreach\i in {0,...,10} \node at (\i,11) {$\vdots$};
	    \foreach\i in {0,...,9} \node at (12,\i) {$\cdots$};

	    \node [draw=none, circle, scale=0.6, fill=none] (s) at (4,5) {};

	    \draw [decorate,decoration={brace,mirror,raise=2pt},black, thick] ($(s)+(0,6)$) -- ($(s)+(-4,6)$)
		    node [midway,yshift=7pt,anchor=south,inner sep=2pt, outer sep=1pt, fill=white]
		    {$\rimDist{x}{s}+1$};		
		    
	    \draw [decorate,decoration={brace,mirror,raise=2pt},black, thick] ($(s)+(4,6)$) -- ($(s)+(1,6)$)
		    node [midway,yshift=7pt,anchor=south,inner sep=2pt, outer sep=1pt, fill=white]
		    {$\rimDist{y}{s}-1$};		
	    
	    \draw[decorate,decoration={brace,mirror,raise=2pt},black, thick] ($(s)+(8.2,1)$) -- ($(s)+(8.2,3)$)
		    node [midway,xshift=7pt,anchor=west,inner sep=2pt, outer sep=1pt, fill=white]
		    {$\rimDist{x}{s}-1$};
		    
	    \draw[decorate,decoration={brace,mirror,raise=2pt},black, thick] ($(s)+(8.2,-5)$) -- ($(s)+(8.2,0)$)
		    node [midway,xshift=7pt,anchor=west,inner sep=2pt, outer sep=1pt, fill=white]
		    {$\rimDist{y}{s}+1$};

	    \path [draw,black,very thick, ->] ($(s)$) -- ($(s)+(0,5)$);
	    \path [draw,black,very thick] ($(s)$) -- ($(s)+(-4,0)$);
	    \path [draw,black,very thick, ->] ($(s)+(-4,0)$) -- ($(s)+(-4,5)$);
	    \path [draw,black,very thick] ($(s)+(0,1)$) -- ($(s)+(-3,1)$);
	    \path [draw,black,very thick, ->] ($(s)+(-3,1)$) -- ($(s)+(-3,5)$);
	    \path [draw,black,very thick] ($(s)+(0,2)$) -- ($(s)+(-2,2)$);
	    \path [draw,black,very thick, ->] ($(s)+(-2,2)$) -- ($(s)+(-2,5)$);
	    \path [draw,black,very thick] ($(s)+(0,3.01)$) -- ($(s)+(-1,3.01)$);
	    \path [draw,black,very thick, ->] ($(s)+(-1,3)$) -- ($(s)+(-1,5)$);
	    
	    \path [draw,black,very thick, ->] ($(s)+(0,1)$) -- ($(s)+(7,1)$);
	    \path [draw,black,very thick, ->] ($(s)+(0,2)$) -- ($(s)+(7,2)$);
	    \path [draw,black,very thick, ->] ($(s)+(0,3.01)$) -- ($(s)+(7,3.01)$);
	    
	    \path [draw,black,very thick, ->] ($(s)$) -- ($(s)+(7,0)$);
	    \path [draw,black,very thick] ($(s)$) -- ($(s)+(0,-5)$);
	    \path [draw,black,very thick, ->] ($(s)+(0,-5)$) -- ($(s)+(7,-5)$);
	    \path [draw,black,very thick] ($(s)+(1,0)$) -- ($(s)+(1,-4)$);
	    \path [draw,black,very thick, ->] ($(s)+(1,-4)$) -- ($(s)+(7,-4)$);
	    \path [draw,black,very thick] ($(s)+(2,0)$) -- ($(s)+(2,-3)$);
	    \path [draw,black,very thick, ->] ($(s)+(2,-3)$) -- ($(s)+(7,-3)$);
	    \path [draw,black,very thick] ($(s)+(3,0)$) -- ($(s)+(3,-2)$);
	    \path [draw,black,very thick, ->] ($(s)+(3,-2)$) -- ($(s)+(7,-2)$);
	    \path [draw,black,very thick] ($(s)+(4,0)$) -- ($(s)+(4,-1)$);
	    \path [draw,black,very thick, ->] ($(s)+(4,-1)$) -- ($(s)+(7,-1)$);
	    
	    \path [draw,black,very thick, ->] ($(s)+(1,0)$) -- ($(s)+(1,5)$);
	    \path [draw,black,very thick, ->] ($(s)+(2,0)$) -- ($(s)+(2,5)$);
	    \path [draw,black,very thick, ->] ($(s)+(3,0)$) -- ($(s)+(3,5)$);
	    \path [draw,black,very thick, ->] ($(s)+(4,0)$) -- ($(s)+(4,5)$);

	    \path [draw, dashed, red,very thick] ($(s)+(4,0)$) -- ($(s)+(5,0)$);
	    \path [draw, dashed, red,very thick, ->] ($(s)+(5,0)$) -- ($(s)+(5,5)$);
	    \path [draw,dashed,red,very thick] ($(s)+(5,0)$) -- ($(s)+(6,0)$);
	    \path [draw,dashed,red,very thick, ->] ($(s)+(6,0)$) -- ($(s)+(6,5)$);
	    \path [draw,dashed,red,very thick, ->] ($(s)+(0,3)$) -- ($(s)+(0,4)$) -- ($(s)+(7,4)$);

	    \draw [orange, very thick] ($(s)+(0,2)$) ellipse (0.3 and 2);
	    \draw [orange, very thick] ($(s)+(3,0)$) ellipse (3 and 0.3);
	    
	    \node [draw, circle, scale=0.6, fill] at (s) {};
	    \node [scale=1] (sLabel) at ($(s)-(0.5,0.5)$) {$s$};

    \end{tikzpicture}
    \subcaption{Path construction at vertex $s$.}%
    \label{fig:path_constr}
  \end{subfigure}%
  \hfill
  \begin{subfigure}[b]{.475\textwidth}
    \centering
    \begin{tikzpicture}[scale=0.46]
	    \path [use as bounding box] (-1.5,-0.5) rectangle (10.5,10.7);
	    \draw [help lines] grid (6.9,6.9);
	    \foreach\i in {0,...,6} \node at (\i,8) {$\vdots$};
	    \foreach\i in {0,...,6} \node at (8,\i) {$\cdots$};

	    \node [draw=none, circle, scale=0.6, fill=none] (s) at (3,3) {};

	    \path [draw,black,very thick] ($(s)$) -- ($(s)+(0,2)$);
	    \path [draw, red, densely dotted, very thick, ->] ($(s)+(0,2)$) -- ($(s)+(0,3)$);

	    \path [draw,blue,dashed,very thick] ($(s) + (0,2)$) -- ($(s)+(-1,2)$);
	    \path [draw,blue,dashed,very thick, ->] ($(s) + (-1,2)$) -- ($(s)+(-1,3)$);

	    \path [draw,black,very thick] ($(s)$) -- ($(s)+(-3,0)$);
	    \path [draw,black,very thick, ->] ($(s)+(-3,0)$) -- ($(s)+(-3,3)$);

	    \path [draw,black,very thick] ($(s)$) -- ($(s)+(2,0)$);
	    \path [draw,red, densely dotted,very thick, ->] ($(s)+(2,0)$) -- ($(s)+(3,0)$);

	    \path [draw,blue,dashed,very thick] ($(s)+(2,0)$) -- ($(s)+(2,-1)$);
	    \path [draw,blue,dashed,very thick, ->] ($(s)+(2,-1)$) -- ($(s)+(3,-1)$);

	    \path [draw,black,very thick] ($(s)$) -- ($(s)+(0,-3)$);
	    \path [draw,black,very thick, ->] ($(s)+(0,-3)$) -- ($(s)+(3,-3)$);

	    \path [draw,black,very thick] ($(s)+(1,0)$) -- ($(s)+(1,1)$);
	    \path [draw,red, densely dotted, very thick, ->] ($(s)+(1,1)$) -- ($(s)+(1,3)$);

	    \path [draw,black,very thick] ($(s)+(1,0)$) -- ($(s)+(1,-2)$);
	    \path [draw,black,very thick, ->] ($(s)+(1,-2)$) -- ($(s)+(3,-2)$);

	    \path [draw,blue,dashed,very thick] ($(s) + (1,1)$) -- ($(s)+(-2,1)$);
	    \path [draw,blue,dashed,very thick, ->] ($(s)+(-2,1)$) -- ($(s)+(-2,3)$);

	    \node [draw, circle, scale=0.6, fill] at (s) {};
	    \node [scale=1] (sLabel) at ($(s)-(0.5,0.5)$) {$s$};

		    \node at ($(s)+(-0.7,6.3)$) {$U_2$};
		    \node at ($(s)+(-3.7,6.3)$) {$U_0$};
		    \node at ($(s)+(-2.2,6.3)$) {$U_1$};

		    \node at ($(s)+(6.3,-1)$) {$R_2$};
		    \node at ($(s)+(6.3,-2)$) {$R_1$};
		    \node at ($(s)+(6.3,-3)$) {$R_0$};
    \end{tikzpicture}
    \subcaption{Filling gaps by rerouting.}%
    \label{fig:path_gap_fill}
  \end{subfigure}
  \caption{Sketched aspects of the construction described in the proof of Lemma~\ref{lem:constr_large_k}. \subref{fig:path_constr} Path construction at vertex $s$; dashed: one shared edge; black: only a shared edge every two paths; orange ellipses enclose shared edges. \subref{fig:path_gap_fill} Filling gaps by rerouting in the construction; dotted: paths before reconstruction, dashed: rerouted paths for filling the gaps.}
  \label{fig:constrs}
\end{figure}
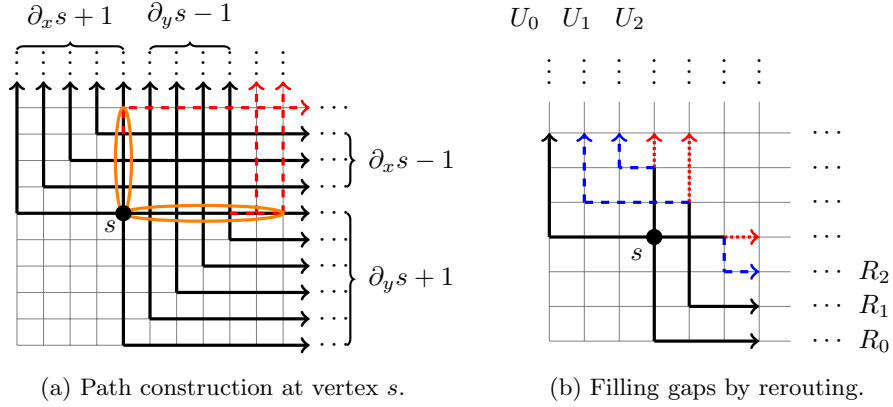

  \begin{proof}
	  From $s$ we construct path fragments (cf.~Fig.~\ref{fig:constrs}) going upwards:
	  \begin{align*}
		  U_{i} &\coloneqq s \uparrow (s_x, s_y + i) \leftarrow (i, s_y + i) \uparrow (i, m - 1 - i), & 0 & \leq i < \rimDist{x}{s},\\
		  U_{i} &\coloneqq s \rightarrow (i, s_y) \uparrow (i, m - 1 - i), & \rimDist{x}{s} & \leq i. %
	  \end{align*}
	  Next, we construct path fragments going to the right:
	  \begin{align*}
		  R_{i} &\coloneqq s \rightarrow (s_x+i,s_y) \downarrow (s_x+i,i) \rightarrow (n - 1 - i,i), & 0 &\leq i < \rimDist{y}{s},\\
		  R_{i} &\coloneqq s \uparrow (s_x, i) \rightarrow (n - 1 - i, i), & \rimDist{y}{s} & \leq i. %
	  \end{align*}
	  To obtain the solution we add the path fragments in the following order. This process is illustrated in Fig.~\ref{fig:path_constr}.
	  \begin{compactenum}[A)]
	  \item We start with $U_0$, $U_{\rimDist{x}{s}}$, $R_0$, and $R_{\rimDist{y}{s}}$ which have no shared edge.
	  This yields $\deg{s}$ paths, since some of these are identical if $s$ lies on the rim.
	  Then for $i=1,\ldots,\rimDist{y}{s}-1$ we add $R_i$ and $U_{\rimDist{x}{s}+i}$, where $R_i$ new adds a shared edge (the other common edges were already shared before). 
	  Afterwards we continue adding $U_i$ and $R_{\rimDist{y}{s}+i}$ for $i=1,\ldots,\rimDist{x}{s}-1$, where $U_i$ adds a shared edge.
	  Thus every other path fragment adds a shared edge. 
	  We stop as soon as we have constructed $p$ paths.

	  \item Continue adding $R_i$ for $i=\rimDist{}{s}+1,\ldots,\rimDist{}{t}$ and $U_j$ for $j=\rimDist{}{s}+1,\ldots,p-\rimDist{}{t}$ until we reach $p$ paths.
	  Here each single fragment adds another shared edge.
\end{compactenum}
	  For $p\leq 2\rimDist{}{s}$ we add the following modifications.
	  If $p\leq 2\rimDist{y}{s}+2$, that is, $R_{\rimDist{y}{s}}$ is the last right-going fragment, extend this fragment downwards such that the endpoints of the $R_i$ form a consecutive line.
	  If the $U_i$ leave a gap, that is, $U_j$ is not part of the construction for some $0<j<\rimDist{x}{s}$, we take the rightmost up-going fragment and route it leftwards along the first free row, and then continue as $U_j$.
	  In the end, if necessary, we extend $U_{\rimDist{x}{s}}$ leftwards, like we did with $R_{\rimDist{y}{s}}$.
	  Thus the endpoints of the up-going fragments form a consecutive line as well.
	  These steps do not introduce further shared edges.
	  See Fig.~\ref{fig:path_gap_fill} for an illustration.

		So in the end we may assume we have constructed fragments $U_0,\ldots,U_{u-1}$ and $R_0,\ldots,R_{r-1}$ for some $u,r\in\N$ with $u+r=p$.
	  At $t$ we proceed analogously, simply mirrored.
	  Therefore we have $r$ down-going fragments $D_i$ and $u$ left-going fragments $L_i$.
	  Then we obtain the solution
	  \begin{align*}
		  \Solution \coloneqq \Set{U_i\cup L_i \mid i=0,\ldots,u} \cup \Set{R_i\cup D_i \mid i=0,\ldots,r}.
	  \end{align*}

	  \noindent\emph{Feasibility.}
	  Furthermore the $R_i$ only use the lower $r$ rows of the grid whereas the $L_i$ use the upper $u$ rows.
	  Since $r+u=p\leq m$, these do not intersect, that is, we do not get further shared edges.
	  The same holds for the $U_i$ and $D_i$, since $p\leq n$.

	  Let $k_s$ and $k_t$ denote the number of shared edges used to construct the path fragments at $s$, and at $t$ respectively.
	  Thus, we have a solution if $k\geq k_s+k_t$.

	  If $p \leq 2(\rimDist{}{s} + 2) -\deg{s}$, then we only use part A. From the $\deg{s}$-th path to the $p$-th path, every other path adds a new shared edge, so $k_s = \Ceil{\inlinefrac{1}{2}(p-\deg{s})}$.
	  Furthermore, $\rimDist{}{s}\leq\rimDistDual{}{t}$, so at $t$ we also only use part A.
	  This implies $k_t = \Ceil{\inlinefrac{1}{2}(p-\deg{t})}$.
	  Hence $k \geq k_s + k_t = \Ceil{\inlinefrac{1}{2}(p-\deg{s})} + \Ceil{\inlinefrac{1}{2}(p-\deg{t})}$.

	  If $2 (\rimDist{}{s} + 2) - \deg{s} < p \leq 2 (\rimDistDual{}{t} + 2) - \deg{t}$, we still only use part A at~$t$ getting $k_t = \Ceil{\inlinefrac{1}{2}(p-\deg{t})}$.
	  But at $s$ we also use part B. 
	  Assume that~$s$ lies in the interior of the grid.
	  Then, when completely executing part A, we use $\rimDist{}{s}-2$ shared edges to construct $2\rimDist{}{s}$ paths.
	  This leaves $k_s-(\rimDist{}{s}-2)$ shared edges for part B. 
	  Each of those allows for another path.
	  So we obtain the condition $p = \rimDist{}{s}+2+k_s$.
	  If $s$ lies on the rim or in the corner, then the argument differs slightly, but the condition is the same.
	  So overall we get the condition $k \geq k_s+k_t = p - (\rimDist{}{s}+2) +\Ceil{\frac{1}{2}(p-\deg{t})}$.
	  
	  Finally, if $p > 2 (\rimDistDual{}{t} + 2) - \deg{t}$, then we use part B at both $s$ and $t$.
	  Thus we have $p = \rimDist{}{s}+2+k_s = \rimDistDual{}{t}+2+k_t$.
	  By adding these equalities we obtain $k_s+k_t = 2p -(\rimDist{}{s} + \rimDistDual{}{t} + 4)\leq k$.
	  So the solution is feasible.

	  \smallskip\noindent\emph{Optimality.}
	  We only give a lower bound for the number $k_s$ of shared edges at $s$.
	  The bound for $k_t$ follows analogously, which then gives the desired bound for $k$.

	  During part~A) of the construction, each contraction may increase the degree of $s$ by at most 2.
	  Hence $p\leq\deg{s}+2k_s$, which shows $k_s\geq\Ceil{\inlinefrac{1}{2}(p-\deg{s})}$.

	  For $p\geq 2(\rimDist{}{s}+2)$ we present a number of cuts of size $p-1$.
	  We use rectangles containing $s$, whose right upper corners move along a diagonal.
	  Formally, these are
	  $\variableStyle{cut}_i \coloneqq \{\{(i,y),(i+1,y)\} \mid y \leq p - 3 - i \} \cup \{\{(x,p-3-i),(x,p-2-i)\} \mid x \leq i\}$
	  for $i=s_x,\ldots,p-3-s_y$. 
	  Assume that~$t$ lies inside one of those rectangles.
	  Then $\variableStyle{cut}_i$ for $s_x\leq i<t_x$ and $p-3-t_y< i\leq p-3-s_y$ are $s$-$t$-cuts of size $p-1$, and these are $\distG{s}{t}$ many.
	  In this case we need $k\geq\distG{s}{t}$ which only allows the trivial solution.
	  So we may assume that~$t$ lies outside all of these rectangles.
	  Thus there are $p-2-\rimDist{}{s}$ many of these cuts and they separate $s$ and $t$.
	  Furthermore they are pairwise disjoint.
	  So we get $k_s \geq p - 2 - \rimDist{}{s}$.

	  Altogether, our construction is optimal.
  \end{proof}
}

\subsection{Holey Grids}
\label{sec:holey}

In the previous section we proved that \MSE{} is solvable in linear time on small and large (compared to the number~$p$ of paths) bounded grids.
In this section we study the complexity of \MSE{} on subgraphs of bounded grids, which we call \emph{holey grids} and show that the problem is $\NP$-hard on this graph class.
To this end we reduce from the well known \Probl{Vertex Cover} problem which is, given a graph $G$ and a natural number $k\in\N$, to decide whether there exists $U \mathop{\subseteq} V$ with $\Abs{U} \mathop{\leq} k$ such that $\forall e \mathop{\in} E\quantify e \mathop{\cap} U\neq \emptyset$. 
More precisely, we use that \Probl{Vertex Cover} remains $\NP$-complete on graphs with maximum degree three~\cite{DBLP:books/fm/GareyJ79}.
Note that our reduction adapts the idea of a reduction used in previous work~\cite{fluschnik2016planar}.

\begin{figure}[t]
	\centering
	\begin{subfigure}[b]{0.15\textwidth}%
	  \centering
	 \begin{tikzpicture}[scale=1.2]
			\node [circle, scale=0.4,fill] (v4) at (0,1) {};
			\node (sLabel) at ($(v4)+(-0.15,+0.15)$) {$v_4$};
			\node [circle, scale=0.4,fill] (v3) at (1,1) {};
			\node (sLabel) at ($(v3)+(+0.15,+0.15)$) {$v_3$};
			\node [circle, scale=0.4,fill] (v2) at (1,0) {};
			\node (sLabel) at ($(v2)+(+0.15,-0.15)$) {$v_2$};
			\node [circle, scale=0.4,fill] (v1) at (0,0) {};
			\node (sLabel) at ($(v1)+(-0.15,-0.15)$) {$v_1$};
			\path [draw] (v1) -- node[auto,swap]{$e_1$} (v2);
			\path [draw] (v2) -- node[auto,swap]{$e_2$} (v3);
			\path [draw] (v3) -- node[auto,swap]{$e_3$} (v4);
			\path [draw] (v1) -- node[auto,swap,label={[label distance=-0.5cm]315:$e_4$}]{} (v3);
		\end{tikzpicture}
	  \subcaption{Example graph.}%
	  \label{fig:AKT-std-graph}
	\end{subfigure}
	\hfill
	\begin{subfigure}[b]{0.37\textwidth}%
	  \centering
	 \begin{tikzpicture}[scale=0.31]
		\draw [color=gray!50!white,ultra thin] (5,4) grid (17,10);
		\node [fill, circle, scale=0.45] (A) at (6,6) {};
		\node [fill, circle, scale=0.45] (B) at ($(A)+(10,0)$) {};
		\path [draw, thick] (A) -- ($(A)+(3,0)$);
		\path [draw, thick] (B) -- ($(B)+(-3,0)$);
		\foreach\i in {1,2,3}
		{
			\node [fill, circle, scale=0.45] at ($(A)+(\i,0)$) {};
			\node [fill, circle, scale=0.45] at ($(B)-(\i,0)$) {};
		}
		\foreach\i in {1,...,5} \node [fill, circle, scale=0.45] at ($(A)+(\i+2,1)$) {};
		\foreach\i in {1,...,4} \path [draw, thick] ($(A)+(\i+2,1)$) -- ($(A)+(\i+3,1)$);
		\foreach\i in {3,7} \path [draw, thick] ($(A)+(\i,1)$) -- ($(A)+(\i,0)$);
		\foreach\i in {1,...,7} \node [fill, circle, scale=0.45] at ($(A)+(\i+1,2)$) {};
		\foreach\i in {1,...,6} \path [draw, thick] ($(A)+(\i+1,2)$) -- ($(A)+(\i+2,2)$);
		\foreach\i in {2,8}
		{
			\path [draw, thick] ($(A)+(\i,2)$) -- ($(A)+(\i,0)$);
			\node [fill, circle, scale=0.45] at ($(A)+(\i,1)$) {};
		}
		\foreach\i in {1,...,9} \node [fill, circle, scale=0.45] at ($(A)+(\i,3)$) {};
		\foreach\i in {1,...,8} \path [draw, thick] ($(A)+(\i,3)$) -- ($(A)+(\i+1,3)$);
		\foreach\i in {1,9}
		{
			\path [draw, thick] ($(A)+(\i,3)$) -- ($(A)+(\i,0)$);
			\node [fill, circle, scale=0.45] at ($(A)+(\i,1)$) {};
			\node [fill, circle, scale=0.45] at ($(A)+(\i,2)$) {};
		}
		\draw[decorate,decoration={brace,mirror,raise=5pt},black, thick] (A) -- ($(A)+(2.9,0)$)
			node [midway,yshift=-7pt,anchor=north,inner sep=2pt, outer sep=1pt, fill=white,opacity=0.3]
			{$M$};
			\draw[decorate,decoration={brace,mirror,raise=5pt},black, thick] (A) -- ($(A)+(2.9,0)$)
			node [midway,yshift=-7pt,anchor=north,inner sep=2pt, outer sep=1pt]
			{$M$};
		\draw[decorate,decoration={brace,raise=5pt},black, thick] (B) -- ($(B)-(2.9,0)$)
			node [midway,yshift=-7pt,anchor=north,inner sep=2pt, outer sep=1pt, fill=white,opacity=0.3]
			{$M$};
			\draw[decorate,decoration={brace,raise=5pt},black, thick] (B) -- ($(B)-(2.9,0)$)
			node [midway,yshift=-7pt,anchor=north,inner sep=2pt, outer sep=1pt]
			{$M$};
		\draw[decorate,decoration={brace,mirror,raise=5pt},black, thick] ($(A)+(3,0)$) -- ($(A)+(7,0)$)
			node [midway,yshift=-7pt,anchor=north,inner sep=2pt, outer sep=1pt, fill=white,opacity=0.6]
			{$\geq k'-1$};
			\draw[decorate,decoration={brace,mirror,raise=5pt},black, thick] ($(A)+(3,0)$) -- ($(A)+(7,0)$)
			node [midway,yshift=-7pt,anchor=north,inner sep=2pt, outer sep=1pt]
			{$\geq k'-1$};
	\end{tikzpicture}
	\subcaption{Rainbow gadget.}%
	\label{fig:rainbow}
	\end{subfigure}
	\hfill
	\begin{subfigure}[b]{0.45\textwidth}
	\begin{tikzpicture}[scale=0.26]
		\draw [color=gray!50!white,ultra thin] (0,0) grid (17,9);

		\node [fill, circle, scale=0.45,label={$v'_{i,j}$}] (x1) at (1,9) {};
		\node [fill, circle, scale=0.45,label={270:$v'_{i+1,j}$}] (y1) at ($(x1)-(0,9)$) {};
		\path [draw, very thick] (x1) -- ++(0,-4) -- ++(6,0) -- ++(0,-1) -- ++(-6,0) -- ++(0,-4);
		
		\node [fill, circle, scale=0.45,label={$v'_{i,j+1}$}] (x2) at (5,9) {};
		\node [fill, circle, scale=0.45,label={270:$v'_{i+1,j+1}$}] (y2) at ($(x2)-(0,9)$) {};
		\path [draw, very thick] (x2) -- ++(0,-3) -- ++(6,0) -- ++(0,-3) -- ++(-6,0) -- ++(0,-3);
		
		\node [fill, circle, scale=0.45,label={$v'_{i,j+2}$}] (x3) at (10,9) {};
		\node [fill, circle, scale=0.45,label={270:$v'_{i+1,j+2}$}] (y3) at ($(x3)-(0,9)$) {};
		\path [draw, very thick] (x3) -- ++(0,-2) -- ++(6,0) -- ++(0,-5) -- ++(-6,0) -- ++(0,-2);
	\end{tikzpicture}
	\subcaption{Snake-chains.}%
	\label{fig:snakes}
	\end{subfigure}
	\begin{subfigure}[b]{1\textwidth}\vspace{10pt}
	  \centering
	  \begin{tikzpicture}[scale=0.35]
		\colorlet{arrowGreen}{green!50!black}
		\node [circle, scale=0.3,fill] (s) at (0,0) {};
		\node (sLabel) at ($(s)+(-0.5,-0.5)$) {$s$};

		\foreach\j in {-1,1}
		{
			\foreach\k in {-1,1}
			{
				\path [draw] ($(s)+(1,0)$) -- ($(s)+(1,0+\j*2)$);
				\path [draw] ($(s)+(0,0)$) -- ($(s)+(1,0)$);
				\path [draw] ($(s)+(2,0+\j*2)$) -- ($(s)+(2,\j*2+\k)$);
				\path [draw] ($(s)+(1,0+\j*2)$) -- ($(s)+(2,\j*2)$);
				\path [draw] ($(s)+(2,0+\j*2+\k)$) -- ($(s)+(3,\j*2+\k)$);
			}
		}

		\path [draw] ($(s)+(3,3)$) -- ($(s)+(3,5)$);
		\path [draw, <-, red, shorten <=1pt] ($(s)+(3.2,3)$) -- ($(s)+(3.2,4.8)$);
		\path [draw, ->, red, shorten >=1pt] ($(s)+(3.2,4.8)$) -- ($(s)+(5,4.8)$);
		\node [inner sep=1pt,label=right:\textcolor{red}{$a$}] (X) at ($(s)+(3.2,4)$) {};
		\path [draw] ($(s)+(3,5)$) -- ($(s)+(5,5)$);

		\def\gridXScale{2.7}
		\def\gridYScale{10/3}
		\def\arrowBegin{5}
		\def\gridWidth{4*\gridXScale}
		\rainbow{($(s)+(\arrowBegin,5)$)}{3}
		\rainbow{($(s)+(\arrowBegin,-5)$)}{3}
		\path ($(s)+(\arrowBegin,-5)$) edge ($(s)+(\arrowBegin-2,-5)$);
		\path ($(s)+(\arrowBegin-2,-5)$) edge ($(s)+(\arrowBegin-2,-3)$);

		\node [circle, scale=0.3,fill] (t) at ($(s)+(2*\arrowBegin+2+\gridWidth+4,0)$) {};
		\node (tLabel) at ($(t)+(0.5,0.5)$) {$t$};		

		\node (gridStartingVertex) at ($(s)+(\arrowBegin+3,-5)$) {};
		\path (s) edge[line width=3pt,green] ($(s)+(0,8)$);
		\path ($(s)+(0,8)$) edge[line width=3pt,green] ($(gridStartingVertex)+(0,3*\gridYScale+3)$);
		\path ($(gridStartingVertex)+(0,3*\gridYScale+3)$) edge[line width=3pt,green] ($(gridStartingVertex)+(0,3*\gridYScale)$);
		\path (t) edge[line width=3pt,green] ($(t)-(0,7)$);
		\path ($(t)-(0,7)$) edge[line width=3pt,green] ($(gridStartingVertex)+(4*\gridXScale,-2)$);
		\path ($(gridStartingVertex)+(4*\gridXScale,-2)$) edge[line width=3pt,green] ($(gridStartingVertex)+(4*\gridXScale,0)$);
		\VCgrid{(gridStartingVertex)}{\gridXScale}{\gridYScale}{4}

		\rainbow{($(gridStartingVertex)+(-2,\gridYScale)$)}{2}	
		\rainbow{($(gridStartingVertex)+(-2,2*\gridYScale)$)}{2}
		
		\gridPath{($(s)+(3,1)$)}{($(gridStartingVertex)+(-2,2*\gridYScale)$)}		
		\gridPath{($(s)+(3,-1)$)}{($(gridStartingVertex)+(-2,\gridYScale)$)}

		\foreach\j in {-1,1}
		{
			\foreach\k in {-1,1}
			{
				\path [draw] ($(t)-(1,0)$) -- ($(t)-(1,0+\j*2)$);
				\path [draw] ($(t)-(0,0)$) -- ($(t)-(1,0)$);
				\path [draw] ($(t)-(2,0+\j*2)$) -- ($(t)-(2,\j*2+\k)$);
				\path [draw] ($(t)-(1,0+\j*2)$) -- ($(t)-(2,\j*2)$);
				\path [draw] ($(t)-(2,0+\j*2+\k)$) -- ($(t)-(3,\j*2+\k)$);
			}
		}

		\rainbow{($(gridStartingVertex)+(4*\gridXScale,\gridYScale)$)}{2}	
		\rainbow{($(gridStartingVertex)+(4*\gridXScale,2*\gridYScale)$)}{2}
		
		\gridPath{($(gridStartingVertex)+(4*\gridXScale,\gridYScale)+(2,0)$)}{($(t)-(3,1)$)}
		\gridPath{($(gridStartingVertex)+(4*\gridXScale,2*\gridYScale)+(2,0)$)}{($(t)-(3,-1)$)}

		\rainbow{($(t)-(\arrowBegin+3,-5)$)}{3}
		\rainbow{($(t)-(\arrowBegin+3,5)$)}{3}
		\path ($(t)-(\arrowBegin,-5)$) edge ($(t)-(\arrowBegin-2,-5)$);
		\path ($(t)-(\arrowBegin-2,-3)$) edge ($(t)-(\arrowBegin-2,-5)$);
		\path ($(t)-(\arrowBegin,5)$) edge ($(t)-(\arrowBegin-2,5)$);
		\path ($(t)-(\arrowBegin-2,3)$) edge ($(t)-(\arrowBegin-2,5)$);

		\path (s) edge ($(s)+(0,8)$);
		\path ($(s)+(0,8)$) edge ($(gridStartingVertex)+(0,3*\gridYScale+3)$);
		\path ($(gridStartingVertex)+(0,3*\gridYScale+3)$) edge ($(gridStartingVertex)+(0,3*\gridYScale)$);
		\path (t) edge ($(t)-(0,7)$);
		\path ($(t)-(0,7)$) edge ($(gridStartingVertex)+(4*\gridXScale,-2)$);
		\path ($(gridStartingVertex)+(4*\gridXScale,-2)$) edge ($(gridStartingVertex)+(4*\gridXScale,0)$);

	\end{tikzpicture}
	\subcaption{Sketch of the graph obtained in the proof of Theorem~\ref{lem:holey_np_hard}.}%
	\label{fig:VC_red_constr}
	\end{subfigure}
	\caption{An exemplified illustration of the construction in the proof of Theorem~\ref{lem:holey_np_hard}. \subref{fig:AKT-std-graph} A graph representing an example instance of \Probl{Vertex Cover}. \subref{fig:rainbow} An illustration of (the grid-embedding of) the rainbow gadget. 
	\subref{fig:snakes} An illustration of snake-chains.
	\subref{fig:VC_red_constr} Sketch of the holey grid constructed in the proof of Theorem~\ref{lem:holey_np_hard}, exemplified with the instance represented by~\subref{fig:AKT-std-graph}. The highlighted path indicates the validation path.}
	\label{fig:fig}
\end{figure}

\begin{theorem}
	\label{lem:holey_np_hard}
	\MSE{} on holey grids is $\NP$-hard.%
\end{theorem}
\begin{proof}
	Given an instance~$\Instance{I}_{\VC} \coloneqq (G = (V,E),k)$ of \Probl{Vertex Cover} with~$\Delta(G)\leq 3$, we compute an equivalent instance~$\Instance{I}_{\operatorname{MSE}} \coloneqq (G'=(V',E'),s,t,p,k')$ of \MSE{} on holey grids in polynomial time.
	We assume that~$\Abs{V}$ is a power of two (otherwise we add isolated vertices until it is).
	Figure~\ref{fig:VC_red_constr} illustrates the graph obtained by applying the following transformation to the graph shown in Figure~\ref{fig:AKT-std-graph}.
	The main part of the construction is a structure we refer to as \emph{meta-grid}.
	The meta-grid encodes the vertex-edge incidence matrix of the original graph.
	We assume that the obtained graph to be embedded as shown in Figure~\ref{fig:VC_red_constr}, which serves as a reference when we use the terms ``\emph{left}'', ``\emph{right}'', ``\emph{up}'', and ``\emph{down}''. For construction purposes, we refer to paths with~$\ell+1$ vertices as \emph{$\ell$-chains} or chains of length~$\ell$. Whenever a chain is added in the construction, all vertices except the two end-vertices are new.
	
	The main component in the construction is a gadget called \emph{rainbow} (cf.~\cite{fluschnik2016planar}), see Figure~\ref{fig:rainbow}.
	Figure~\ref{fig:rainbow} also shows that this gadget is a subgraph of a bounded grid.
	We use rainbow gadgets where the number of vertices in each band in the spectrum of the rainbow is larger than the number of allowed shared edges. %
	This allows the rainbow gadget to restrict the number of paths that can be routed through it to at most the number of bands in the spectrum.
	Note that in any rainbow that is satiated with~$M$ paths~$2M-2$ edges are shared.
	We call the number of shared edges in a rainbow the \emph{rainbow-offset}.
	
	We define~$M$ and a few other values we need in order to build the graph $G'$ in the following:
	\begin{align*}
		M &\coloneqq 2(\Abs{E}+1)+2; \ \ \ \trees \coloneqq 2\cdot (\Abs{V}\cdot \log_2(\Abs{V}) - 2 + 2k); \\
		c &\coloneqq 10; \ \ \ c' \coloneqq  2\Abs{V} + \Abs{E} \cdot \Abs{V} - 2\Abs{E}; \ \ \ b \coloneqq 2\cdot M\cdot c' + 1;\\
		a_0 &\coloneqq \frac{\Abs{V}-1}{2}(M+c-2)-\log{\Abs{V}}; \ \ \ a \coloneqq \max(a_0,\Abs{E}^3,b^2).
	\end{align*}
	Here,~$c'$ is the number of rainbow gadgets we construct.
	The values~$a$,~$a_0$, and~$b$ are chosen to ensure certain constraints when routing paths and sharing edges while~$c$ can be understood as a scaling constant used to avoid intersections. 
	Why the values are chosen in this way will become clear later in the proof.
	Next we set $p \coloneqq k \cdot M + (\Abs{V} - k) +1$ for the number of paths and $k' \coloneqq k\cdot(2a+b\Abs{E}) + \trees + c'(2M-2)$ for the number of shared edges in $\Instance{I}_{\operatorname{MSE}}$. 
	
	In the following we describe the construction of the meta-grid. 
	First, we create a grid of vertices, without any edges, that has~$\Abs{V}$ rows and~$\Abs{E}+1$ columns. We fix an ordering $v_1,\dots,v_{\Abs{V}}$ on the set~$V$ of vertices and use it to identify each row of the grid with a vertex from~$G$. Analogously, we fix an order $e_1,\dots,e_{\Abs{E}}$ on the edge set~$E$ and use it to identify each space between two consecutive columns of the grid with an edge in~$G$. From here on we will refer to these spaces as columns. %
	
	The first vertex in row~$i$ is denoted~$v'_{i,1} \in V'$, refer to Figure~\ref{fig:VC_red_constr}, the second one is denoted~$v'_{i,2}$, and so on. 
	If vertex~$v_i$ is incident to edge~$e_j$ in~$G$, then vertices~$v'_{i,j}$ and~$v'_{i,j+1}$ are connected by a chain of length~$b$.
	If~$v_i$ is not incident to~$e_j$ in~$G$, then vertices~$v'_{i,j}$ and~$v'_{i,j+1}$ are connected by a chain of length~$b$ followed by a rainbow.
	This completes the construction of the rows.
	
	We embed the structure we just created in a grid such that the first vertices are vertically aligned and have vertical distance of~$M + c$.
	Now we connect each vertex~$v'_{i,j}$ with~$i<\Abs{V}$ and~$j<\Abs{E}$ with its respective lower neighbour, that is, vertex~$v'_{i+1,j}$, by so-called \emph{snake-chains} of length at least $k' + 1$ (the wavy vertical lines in Figure~\ref{fig:VC_red_constr}).
	Note that these vertices do not necessarily lie above each other.
	The snake-chains are constructed as follows (refer to \Cref{fig:snakes} for an illustration).

	In every row except the lowest one, we start with the left most snake-chain.
	We first route it four steps down, then~$k'$ steps to the right, one down, left again until we are above its end-vertex which we then join it to by a vertical path.
	Then every further snake-chain is routed the following way: down by the maximum possible number of steps (at most four) such that no previous snake-chain is intersected, then~$k'$ to the right, then the minimum necessary number of steps down, such that the snake-chain can be extended to the left without intersecting a previous snake-chain until it can be routed downwards until it meets its end-vertex.

	Note that the above description implies that we reduce the number of steps that a snake-chains is routed downwards every time the previous column did not contain a rainbow.
	After a rainbow is encountered we start with four steps down again.
	Since $G$ has a maximum degree of three, there are at most three columns in every row without a rainbow, so after at most four consecutive snake-chains we encounter a rainbow in the next column.
	This way the snake-chains do not intersect or touch each other and the constant $c > 2\cdot 4 +1$ ensures that they also do not intersect any rainbows from the next row.

	Now we add a source vertex $s$ to the left of the meta-grid and construct a complete, binary tree of height~$\log_2{\Abs{V}}$ with~$s$ as its root and with~$\Abs{V}$ leaves pointing in direction of the grid.
	We construct this tree in such a way that all vertices of the same level lie in the same column of the grid and from one leaf to the next we have distance two in the grid.
	This is possible since the number of vertices in~$G$ is a power of two.
	To make this tree embeddable into a grid we replace every edge by a chain of the minimal required length running along the grid structure.
	We connect the uppermost leaf to the first row of the meta-grid in a way such that the vertical distance between this leaf and~$v'_{1,1}$ is exactly~$a_0$.
	More specifically, we add a chain up and to the right until it has length~$a$, then add a rainbow of sufficient length and connect it to~$v'_{1,1}$.
	Each leaf of the tree is connected by a chain of length~$a$ and a following rainbow to one of the vertices in the first column of the meta-grid such that the order of the leaves and the vertices is the same.
	The length~$a$ is chosen such that all the chains have the same length.
	To avoid intersections in the $a$-chains these go right first: the chain leading to the row corresponding to~$v_i$ is routed~$i-1$ steps to the right if~$i < \frac{\Abs{V}}{2}$ and $\Abs{V}-i$ otherwise.
	Then the chains go up/down to their row and then right until they have length~$a$.
	Note that this tree is symmetrical in the end since we work on an even number of vertices.
	
	The same is done on the right side: we add a vertex~$t$ and a binary tree to its left with~$t$ being the root and the leaves are connected to the vertices in the last column of the meta-grid by a chain of length~$a$ and a rainbow.
	If the construction of the snake-paths causes some of the snake-paths to ``stick out'' to the right, then we extend the paths in the rainbows at the leaves of~$t$ as far as necessary to ensure that nothing intersects.
	The length of these rainbows is also used to align the leaves of the tree on this side.
	
	Finally, we add chains of length at least~$k'+1$, the \emph{outer-grid chains}, one connecting~$s$ to~$v'_{1,1}$ and the other connecting~$v'_{\Abs{V},\Abs{E}+1}$ to~$t$. 
	
	Intuitively, the correctness is shown as follows. 
	Recall that~$p \coloneqq k \cdot M + (\Abs{V} - k) +1$. 
	We know that we can route at most~$M$ paths through a rainbow, this we have to do $k$ times. 
	So we can pick~$k$ of the $\Abs{V}$ rows and route $M$ paths through each. 
	We route a single path through each of the remaining~$\Abs{V}-k$ rows. 
	Now we have to route one additional path, which has to use the outer-grid chains and the snake-chains. 
	This path will verify that the~$k$ rows we chose to route~$M$ paths through correspond to vertices of~$G$ that constitute a vertex cover. 
	Then each column corresponding to an edge of~$G$ has at least one row where we have a fully shared chain and no rainbow. 
	So the remaining path can be routed through those chains and use the snake-chains to switch between rows. 
	Of course,~$k'$ is chosen in a way that we are forced to use the described approach and that there is no solution if $G$ does not have a vertex cover of size~$k$.
	We claim that $G$ has a vertex cover of size $k$ if and only if $G'$ has $p$ paths from $s$ to $t$ sharing at most $k'$ edges.
	
	\raproof{} 
	Let~$\VC \subseteq V$ be a vertex cover of size at most~$k$. Without loss of generality we assume~$\Abs{\VC} = k$.
	Then for each~$v \in \VC$, we route~$M$ paths from~$s$ via the chain of length~$a$ and through the following rainbow leading to the corresponding vertex~$v'$.
	This way each~$M$ paths cause~$a$ shared edges within the corresponding chain plus the rainbow-offset of~$2M-2$ edges, so in total~$k \cdot (a + 2M-2)$.
	For all other vertices we route only one path this way.
	Doing so we cause every edge in the tree on the side of~$s$ to be shared except on its lowest level, where only~$k$ branches of the tree are shared, notice that this yields exactly~$\trees/2$ shared edges.
	So the paths cause $a \cdot k + k (2M-2) + \trees/2$ shared edges before reaching the meta-grid.
	Next we route all these paths horizontally through our meta-grid yielding another $k \cdot b \cdot \Abs{E} + k \cdot (2M-2)$ shared edges.

	To route the paths from the meta-grid to~$t$ we get additional $k \cdot (a + 2M-2) + \trees/2$ shared edges.
	So we have $k \cdot M + (\Abs{V} - k)$ paths sharing $k\cdot(2a+b\Abs{E}) + \trees + c'(2M-2) \leq k'$ edges.
	The number~$\trees$ therefore describes exactly the number of edges in the trees at~$s$ and~$t$ that are shared.
	
	Next we route one additional path~$P$ without sharing any additional edges.
	Starting at~$s$ we route~$P$ along the outer-grid chain linking~$s$ to~$v'_{1,1}$, the first vertex in the first row.
	From there~$P$ has to pass the columns of the meta-grid.
	Since~$\VC$ is a vertex cover there is a covering vertex for every edge.
	So for every column we take the vertex covering it and route the path to the corresponding row using the snake-chains.
	Since the vertex is part of the vertex cover the chain crossing this column is already shared, so we can use it to route~$P$ on to the next column.
	After~$P$ has crossed the meta-grid in this fashion it can be routed to the lowest row via the snake-chains and then via the outer-grid chain leading to~$t$.
	
	So~$G'$ allows for $k \cdot M + (\Abs{V}-k)+1$ $s$-$t$-paths sharing $k' \leq k\cdot(2a+b\Abs{E}) + \trees + c'(2M-2)$ edges. 
	Hence,~$\Instance{I}_{\operatorname{MSE}}$ a yes-instance.
	
	\laproof{}
	Assume that~$G'$ has a solution~$\Solution$, a set of $s$-$t$-paths with~$\Abs{\Solution} = p$ sharing at most~$k'$ edges.
	First consider how the paths leave~$s$.
	There are~$|V|+1$ ways to get from~$s$ to the meta-grid.
	One of those is the outer-grid chain which can only contain one path since it has length~$k'+1$.
	So~$p-1 = k \cdot M + (\Abs{V}-k)$ paths have to be routed through the tree.
	Also every rainbow can contain at most~$M$ paths.
	And due to~$k'<2a(k+1)$ at most~$2k$ of the $a$-chains in~$G'$ can be shared, which means~$k$ of the $a$-chains connected to~$s$ and~$k$ of the $a$-chains connected to~$t$.
	Consequently, we get that~$k$ of the rainbows connected to the first column of the meta-grid contain~$M$ paths each.
	Let~$v'_{i_1 ,1},\dots,v'_{i_k,1}$ be the first vertices from these rows.
	The other rainbows connected to the first column then contain exactly one path.
	
	Next we consider what happens inside the meta-grid.
	Since the snake-chains are the only way to leave a row and those cannot be shared, at most~$2(\Abs{E}+1)$ paths can leave a row.
	So in the last column at least 2 paths are still routed through the rainbow connected to these rows and then via the $a$-chain to~$t$.
	We can do the same maths as before on the~$t$ side to get that these rainbows have to contain~$M$ paths as well.
	Hence, each of this~$k$ rows induces at least~$2a+b \Abs{E}$ shared edges.
	Also~$\trees$ many shared edges are induced within the trees by construction.
	Note that~$b > c' \cdot (2M-2)$.
	Hence, there are less than~$b$ shared edges left.

	Now let $\VC \coloneqq \Set{v_{i_1},\dots,v_{i_k}}$.
	Assume towards a contradiction that~$\VC$ is not a vertex cover.
	Then there is a column with a rainbow in every of the~$k$ rows corresponding to the vertices in~$\VC$.
	That means that the path that is routed via the outer-grid chain connected to~$s$ cannot pass this column without sharing additional~$b$ edges, but our budget does not suffice for this.
	This yields a contradiction to~$G'$ being a yes-instance.
\end{proof}
\subsection{Manhattan-like Acyclic Digraphs}
\label{sec:manhattan}

In the previous section, we proved that~\MSE{} is $\NP$-hard on holey grids, i.e.\ subgraphs of a bounded grid.
Along the line, in this section we prove that the \emph{directed} version, \DMSE{}, is $\NP$-hard on the graph class of acyclic directed holey grids (we refer to this class by~\emph{Manhattan {DAG}s}).
We remark that inspired by the street design of Manhattan, New York City, directed bounded grids (referred to as \emph{Manhattan street networks}) are considered in the literature, also in the context of routing~\cite{maxemchuk1987routing,Varvarigos98}.

Observe that \MSE{} reduces to \DMSE{} by replacing each edge~$\{u,v\}$ by anti-parallel arcs~$(u,v)$, $(v,u)$.
The correctness of the reduction follows immediately from the following.

\begin{lemma}%
  \label{lem:not_both_dir}
  Let $(G,s,t,p,k)$ be an instance of \DMSE{}.
  If $\Solution$ is a solution for this instance where two paths~$P_A$ and~$P_B$ use $e = (u,v) \in E$ and its inverted arc $e' = (v,u) \in E$, then we can find a solution $\Solution'$ for the same instance that does not use both of these arcs.
\end{lemma}

\appendixproof{Lemma}{lem:not_both_dir}%
{
  \begin{proof}
  The idea is to split both paths by removing those two edges and connect the beginning of the first path to the trailing part of the second path and vice versa, retaining a solution for the instance.
  
  We first introduce some notation: When $P$ is a directed path and $u,v \in P$ occur in that order, then we denote the subpath of $P$ that starts at $u$ and ends at $v$ by $P[u,v]$.
  If $P_1$ and $P_2$ are two paths and the last vertex of $P_1$ is the same as the first vertex of $P_2$, then we write $P_1 \cdot P_2$ for the path resulting by the union of the two paths and removing all cycles.
  
  Assume $(G,s,t,p,k)$ has a solution $\Solution$ with $e\in \Path_A\in\Solution$ and $e'\in\Path_B\in\Solution$.
  Then we can split $P_A$ and $P_B$ into subpaths as follows: $P_A = P_A[s,u] \cdot e \cdot P_A[v,t]$ and $P_B = P_B[s,v] \cdot e' \cdot P_B[u,t]$.

  Now, by replacing path $P_A$ by $P_A' \coloneqq P_A[s,u] \cdot P_B[u,t]$ and $P_B$ by $P_B' \coloneqq P_B[s,v] \cdot P_A[v,t]$,
  we are able to retain a solution in which these two paths do not use $e$ and $e'$ any more.

  We can repeat this until one of the two edges is not used any more by any of the paths in the solution, yielding the desired $\Solution'$.
  \end{proof}
}

However, the directed graph obtained in the reduction is not acyclic.
We show next that \DMSE{} remains hard even on acyclic directed holey grids.
On a high level, we adapt the construction presented in the proof of Theorem~\ref{lem:holey_np_hard}.
We then direct the edges from left to right, from~$s$ towards~$t$.
Finally, we duplicate the horizontal chains (snake chains) and direct one upwards and one downwards.

\begin{theorem}%
  \label{thm:manhattandags}
	\DMSE{} on Manhattan {DAG}s is~$\NP$-hard.
\end{theorem}
\appendixproof{Theorem}{thm:manhattandags}
{
\begin{proof}
We adapt the reduction for holey grids given in the proof of Theorem~\ref{lem:holey_np_hard}.
To obtain a directed graph, we replace every edge in the corresponding construction by a directed edge:
First we ignore the vertical snake-chains and handle the remaining edges.
Each horizontal edge is directed to the right, rainbows go up, then right, then down.
The tree edges are directed away from $s$ and towards $t$.
Because each snake-chain has to offer both vertical directions in the new construction without creating a cycle, they are constructed as follows.
We replace the undirected snake-chains of the original construction by two snake-chains separated by single edges in their start and end.
The single edges are directed to the right, the left of the new snake-chains is directed downwards and the right one upwards.
We call a snake-chain outgoing from the row in which it starts and ingoing to the row where it ends.
To ensure that we still have enough space for the additional snake-chain, we increase the distance between the rows by defining $c:=20$.

Let $\Instance{I}_{\VC} \coloneqq (G,k)$ be an instance of \Probl{Vertex Cover} and $\Instance{I}_{\operatorname{MSE}}' \coloneqq (G',s,t,p,k')$ be an instance of \MSE{} constructed using the reduction for holey grids.
From this we construct an instance $\Instance{I}_{\operatorname{DMSE}}'' \coloneqq (G'',s,t,p,k'')$ of \DMSE{} according to the additional steps described above.
In the construction for holey grids there are chains of length~$b$ inside the meta-grid.
Due to the additional single edges connecting each pair of snake-chains, we get chains of length~$b' \coloneqq b+1$ in this construction.
Because $b$~is used in the definition of~$k'$, the incrementation of~$b$ leads to~$k'' \coloneqq k' + k \cdot \Abs{E}$.

    We claim that $G$~has a vertex cover of size $k$ if and only if $G''$ has $p$ paths from $s$ to $t$ sharing at most $k''$ edges.
    
    \raproof{}
    This direction of the proof works analogously to the reduction for holey grids using the new instance $\Instance{I}_{\operatorname{DMSE}}''$ instead of $\Instance{I}_{\operatorname{MSE}}'$.
    The first $k \cdot M + (|V|-k)$ paths are routed in the same way. %
    Then again we route one additional path by always choosing the corresponding outgoing snake-chain.
    Since we have a vertex cover, this does not share further edges.

    \laproof{}
    Let $\Solution$ be a solution for $\Instance{I}_{\operatorname{DMSE}}''$.
    The difference to the reduction for holey grids lies in the number of snake-chains that is doubled in the new construction.
    Now, after sharing $\trees$ edges (for definition of $\trees$ see the proof of \Cref{lem:holey_np_hard}) in the binary trees at $s$ and $t$, at most $2k$ of the $a$-chains between the binary trees and the meta-grid in $G''$ can be shared.

    Since only the outgoing snake-chains can be used to leave a row, we can argue again that there are $k$ rows in the meta-grid containing more than one path, each one inducing at least $2a+b' \Abs{E}$ shared edges.
    Let $\VC$ be the set vertices corresponding to these rows.
    If $\VC$ is no vertex cover, then we need to share another $b'$-chain in the construction adding additional $b'$ shared edges, which is a contradiction to $\Instance{I}_{\operatorname{DMSE}}''$ being a yes-instance.
  \end{proof}
}

\section{The Nonexistence of Polynomial Kernels}
\label{sec:polyKernel}
\appendixsection{\ref{sec:polyKernel}}

In this section, we consider \MSE{} from a parametrised complexity point of view.
\MSE{} is kernelisable but does not admit a polynomial problem kernel when it is parametrised by the number~$p$ of paths, unless $\NP \subseteq \coNPPoly$~\cite{fluschnik2015para_MSE}.
We strengthen the latter result and complement the intractability of \MSE{} on planar graphs by showing the following.

\begin{theorem}\label{thm:noPK}
	Unless $\NP \subseteq \coNPPoly$, \MSE{} with parameter $\kappa(G,s,t,p,k) \coloneqq p+k+\Delta(G)+\diamG{G}+\tw{(G)}$ does not admit a polynomial kernel, even on planar graphs.
\end{theorem}

In order to prove Theorem~\ref{thm:noPK}, we use a so-called \OrCross{} due to Bodlaender et al.~\cite{bodlaender2014kernelization}. %
Therein, one uses a \emph{polynomial equivalence relation} $\polyEquivRel$ which is an equivalence relation that is decidable in polynomial time and for each finite set $S$, the number of equivalence classes with respect to $\polyEquivRel$, that is, $\Abs{\Set{[s]_{\polyEquivRel}\mid s \in S}}$, is polynomially bounded in the size of the largest element in $S$.

\begin{definition}[\OrCross{}~\cite{bodlaender2014kernelization}]
	Let $L \mathop{\subseteq} \Sigma^*$ be some problem and $Q \mathop{=} (P,\kappa)$ with $P \mathop{\subseteq} \Sigma^*$ and $\kappa \mathop{:} \Sigma^* \mathop{\to} \N$ be some parametrised problem.
	Furthermore, let $\polyEquivRel$ be a polynomial equivalence relation on $\Sigma^*$.
	An \emph{\OrCross{}} is an algorithm that gets instances $\Instance{I}_1,\ldots,\Instance{I}_q$ of $L$ as input, all of them belonging to the same equivalence class of $\polyEquivRel$,
	and outputs an instance $\Instance I$ of $Q$ such that
	\begin{compactitem}
		\item $\Instance I\in P$ if and only if there is at least one $i$ such that $\Instance I_i\in L$ and
		\item $\kappa(\Instance I)$ is polynomially bounded in $\max\Set{\Abs{\Instance I_i}\mid i=1,\ldots,q} + \log q$.
	\end{compactitem}
\end{definition}

If there is an \OrCross{} from an $\NP$-hard problem~$L$ to some parametrised problem~$Q$, then $Q$ does not admit a polynomial-size kernel, unless~$\NP \subseteq \coNPPoly$~\cite{bodlaender2014kernelization}.
Using this result, we give an \OrCross{} to prove Theorem~\ref{thm:noPK}. 
Our construction contains binary trees and we use the following structural result on binary trees with respect to~\MSE{}.

\begin{lemma}%
	\label{lem:tree_choice}
	Let $T$ be a balanced, binary and complete tree of height $h$ with root~$s$, where additionally all leaves are identified with the target $t$.
	Then the \emph{only} solutions for an \MSE-instance $(T,s,t,p,k)$ with $p\geq h+3$ and $k\leq h$ are to share a complete path from $s$ to some leaf, which is only possible for $k=h$.
\end{lemma}
\appendixproof{Lemma}{lem:tree_choice}
{
  \begin{proof}
      The idea is again that the contraction of the shared edges must allow a flow of value $p$ from $s$ to $t$ in any solution.

      Initially $s$ has degree two.
      Each contraction of an edge may only increase the degree by one.
      So after $k$ contractions the degree is at most $h+2 < p$, except if the contraction identifies $s$ with $t$.
      This can only happen if a complete path from~$s$ to~$t$ is shared.
      Since such a path has length $h$ this is only possible for $k \geq h$.
      Due to the condition $k \leq h$ there are no other shared edges and $k = h$.
  \end{proof}
}
Next we prove the main result of this section.

\begin{proof}[Proof of Theorem~\ref{thm:noPK}]
	We apply the OR-cross-composition framework with \MSE{} on planar graphs where $s$ and $t$ lie on the outer face as input problem.
	The $\NP$-hardness of this problem is shown in Theorem~\ref{lem:holey_np_hard} since in the reduction $s$ and $t$ are on the outer face.

	We say an instance~$(G,s,t,p,k)$ of~\MSE{} is \emph{malformed} if $\distG{s}{t}\leq k$ (trivial yes-instances), if $s$ and $t$ are not connected, if $p \geq 2\cdot |E(G)|$ and $k < \distG{s}{t}$ (trivial no-instances), or if $p\leq 2$.
	Note that in the last case we can decide the instance in polynomial time, since the problem is fixed-parameter tractable with respect to $p$~\cite{fluschnik2015para_MSE}. 
	Hence we can decide each malformed instance in polynomial time.

	We define the equivalence relation $\polyEquivRel$ as follows: two instances~$(G,s,t,p,k)$ and~$(G',s',t',p',k')$ are $\polyEquivRel$-equivalent if both are malformed or if~$p=p'$ and~$k=k'$. 
	Observe that $\polyEquivRel$ is a polynomial equivalence relation.
	
	Let $\Instance{I}_i = (G_i,s_i,t_i,p,k)_{1 \leq i \leq q}$ be non-malformed $\polyEquivRel$-equivalent instances of \MSE{}.
	We assume~$q$ to be a power of 2 (as otherwise we duplicate instances until it is).
	We first construct a complete binary tree $T_s$ rooted in $s$ with depth $\log(q)$ such that the~$s_i$ are the leaves of $T_s$, occurring in their canonical order. 
	Conversely, we construct a tree $T_t$ with root $t$ and leaves~$t_i$.
	We subdivide each edge in $T_s$ and $T_t$ to obtain paths of length $k+1$. 
	In this way we obtain a new graph~$G=(V,E)$ with~$V \coloneqq  V(T_s) \cup V(T_t) \cup \bigcup_{i = 1}^{q} V(G_i)$ and~$E \coloneqq  E(T_s) \cup E(T_t) \cup \bigcup_{i = 1}^{q} E(G_i)$.
	Furthermore, we define the new parameters $p'\coloneqq p+\log(q)$ and $k'\coloneqq 2\log(q)\cdot (k+1) + k$ and get the instance $\Instance{I} \coloneqq (G,s,t,p',k')$, see Fig.~\ref{fig:orCross}.
	Now we claim that $\Instance{I}$ is a yes-instance if and only if there is an $\Instance{I}_y$ with $1 \leq y \leq q$ that is a yes-instance.
	
	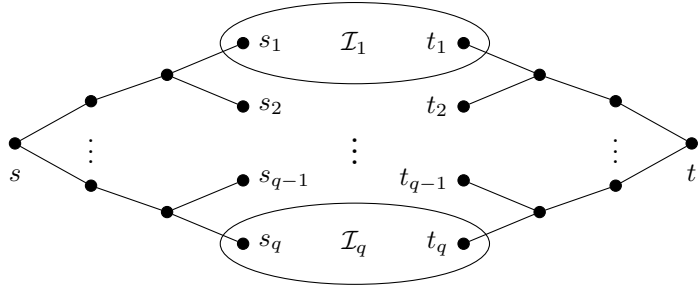
\begin{figure}[t]
	\centering
	\begin{tikzpicture}[yscale=0.7]

		\node (I1Label) at (-0.53,0.6) {$\Instance{I}_q$};
		\node[draw,ellipse,text width=15ex,text height=1.5em] (I1) at ($(I1Label)$) {};
		\node (I2Label) at (-0.53,4.4) {$\Instance{I}_1$};
		\node[draw,ellipse,text width=15ex,text height=1.5em] (I2) at ($(I2Label)$) {};
		\node (dots) at (-0.53,2.5) {\Large{\vdots}};

		\node[draw,circle,fill,scale=0.45,label={[yshift=-0.7cm]$s$}] (s) at (-5,2.5) {};
		\node (dotsT) at ($(s)+(1,0)$) {\vdots};
		\node (sl) [draw,circle,fill,scale=0.45] at ($(s)+(1,0.8)$) {};
		\node (sr) [draw,circle,fill,scale=0.45] at ($(s)+(1,-0.8)$) {};
		\node (s2l) [draw,circle,fill,scale=0.45] at ($(sl)+(1,0.5)$) {};
		\node (s2r) [draw,circle,fill,scale=0.45] at ($(sr)+(1,-0.5)$) {};
		\node (s3ll) [draw,circle,fill,scale=0.45,label=right:$s_1$] at ($(s2l)+(1,0.6)$) {};
		\node (s3lr) [draw,circle,fill,scale=0.45,label=right:$s_2$] at ($(s2l)+(1,-0.6)$) {};
		\node (s3rl) [draw,circle,fill,scale=0.45,label=right:$s_{q-1}$] at ($(s2r)+(1,0.6)$) {};
		\node (s3rr) [draw,circle,fill,scale=0.45,label=right:$s_q$] at ($(s2r)+(1,-0.6)$) {};
		\foreach \i in {l,r}
		{
			\draw (s) edge (s\i);
		}
		\foreach \i in {l,r}
		{
			\draw (s\i) edge (s2\i);
		}	
		\foreach \i in {l,r}
			\foreach \j in {l,r}
			{
				\draw (s2\i) edge (s3\i\j);
			}	

		\node[draw,circle,fill,scale=0.45,label={[yshift=-0.7cm]$t$}] (t) at (3.9,2.5) {};
		\node (dotsS) at ($(t)-(1,0)$) {\vdots};
		\node (tl) [draw,circle,fill,scale=0.45] at ($(t)-(1,0.8)$) {};
		\node (tr) [draw,circle,fill,scale=0.45] at ($(t)-(1,-0.8)$) {};
		\node (t2l) [draw,circle,fill,scale=0.45] at ($(tl)-(1,0.5)$) {};
		\node (t2r) [draw,circle,fill,scale=0.45] at ($(tr)-(1,-0.5)$) {};
		\node (t3ll) [draw,circle,fill,scale=0.45,label=left:$t_q$] at ($(t2l)-(1,0.6)$) {};
		\node (t3lr) [draw,circle,fill,scale=0.45,label=left:$t_{q-1}$] at ($(t2l)-(1,-0.6)$) {};
		\node (t3rl) [draw,circle,fill,scale=0.45,label=left:$t_2$] at ($(t2r)-(1,0.6)$) {};
		\node (t3rr) [draw,circle,fill,scale=0.45,label=left:$t_1$] at ($(t2r)-(1,-0.6)$) {};
		\foreach \i in {l,r}
		{
			\draw (t) edge (t\i);
		}
		\foreach \i in {l,r}
		{
			\draw (t\i) edge (t2\i);
		}	
		\foreach \i in {l,r}
			\foreach \j in {l,r}
			{
				\draw (t2\i) edge (t3\i\j);
			}	
		
	\end{tikzpicture}
	\caption{Construction of $\Instance{I}$ via an \OrCross. The instances are connected by complete binary trees with roots $s$ and $t$, respectively.}
	\label{fig:orCross}
	\end{figure}
	In the trees let $P_x$ denote the path from $s$ to $s_x$ in $T_s$ and let $Q_x$ denote the path from $t_x$ to $t$ in $T_t$ for every $1\leq x\leq q$.

	\laproof{} Assume that $\Instance{I}_y$ is a yes-instance with solution $\Solution_y$.
	We route $p$ paths along $P_y$, through $\Instance{I}_y$ and along $Q_y$ sharing at most $k$ edges within $\Instance{I}_y$ and additional $2\log(q)\cdot(k+1)$ edges in the trees.
	Thus in $\Instance{I}$ at most $k'$ edges are shared.
	Note that every other instance $\Instance{I}_i$ allows for at least one path, since~$s$ and~$t$ are connected.
	So from each vertex of~$P_y$, except $s_y$, we route an additional path through one of the remaining instances.
	Thus we get $\log(q)$ additional paths not sharing any additional edge.
	So we have $p'$ paths and therefore a solution for $\Instance{I}$.

	\raproof{} Assume the constructed instance $\Instance{I}$ has a solution $\Solution$.
	Note that in $T_s$ and $T_t$ we share at most $2\log(q)$ of the $(k+1)$-paths, but we have to route $p'\geq \log(q)+3$ paths, since $p\geq 3$.
	By Lemma~\ref{lem:tree_choice} this implies that there exist $1 \leq y,x \leq q$ such that $P_y$ and $Q_x$ are completely shared.
	At each vertex of~$P_y$ only one path from $\Solution$ may branch off, which implies that only $\log(q)$ paths do so.
	Hence, at least $p$ paths are routed through $P_y$, leaving the instance at $t_y$.
	So $Q_y$ must be shared, which implies $y=x$.
	Note that $\Solution$ therefore shares $2\log(q)(k+1)$ outside of $\Instance{I}_y$.
	That leaves $k = k'-2\log(q)(k+1)$ edges that can be shared inside of $\Instance{I}_y$.
	So $\Solution$ restricted to $\Instance{I}_y$ is a solution for $\Instance{I}_y$.
	
	This concludes the proof that $\Instance{I}$ has a solution if and only if there is an $1 \leq i \leq q$ such that $\Instance{I}_i$ has a solution.
	Finally, we observe that $p'$ and $k'$ are polynomially bounded in $p+\log(q)$ and $k+\log(q)$.
	Since we only added binary trees, the maximum degree $\Delta$ is increased by at most two.
	For all $1 \leq x,y \leq q$ and every two vertices in the instances $\Instance{I}_x$ and $\Instance{I}_y$ there is a connecting path via $s_x$-$s$-$s_y$ of length at most $\diamG{G_x} + 2\log(q) + \diamG{G_y}$.
	For all vertices in $T_s$ and $T_t$ we have a connecting path by going to $s$, then through $\Instance{I}_1$ and finally via $t$ to the desired vertex.
	This path has length at most $4\log(q) + \diamG{G_1}$.
	Hence, the diameter of $G$ is at most $4\log(q)+\max\{\diamG{G_i} \mid i=1,\ldots,q\}$, which is polynomial in the input size and $\log(q)$.
	The treewidth of $G$ is upper-bounded by~$3\cdot \diamG{G}$, because the graph is planar \cite{robertson1984GMIII}. 
	(Alternatively, there also is a tree decomposition of $G$ of treewidth at most $2\log(q) + \max\{\operatorname{tw}(G_i):i=1,\ldots,q\}$.)

	So $\kappa(\Instance I)$ is polynomially bounded by $\max\{\kappa(\Instance I_i):i=1,\ldots,q\}+\log(q)$.
	It follows that \MSE{} parametrised by $k + p + \Delta + \diamG{} + \tw$ does not admit a polynomial kernel, even on planar graphs.
\end{proof}

Recall that \DMSE{} is $\NP$-hard on planar acyclic digraphs with~$s$ and~$t$ lying on the outerface~(Theorem~\ref{thm:manhattandags}). 
Hence, replacing the input instances by instances from \DMSE{} on the aforementioned graphs, and directing the remaining edges in the trees away from~$s$ and towards~$t$ allows us to also exclude polynomial kernels for \DMSE{} parametrised\footnote{\DMSE{} is in $\FPT$ when parametrised by~$p+k$ since the search tree algorithm solving \MSE{} in $O((p-1)^k\cdot (\Abs{V}+\Abs{E})^2)$ time~\cite{fluschnik2015para_MSE} can easily be adapted to the directed case.} by~$p+k+\Delta_{\rm in}(G)+\Delta_{\rm out}(G)$. 

\begin{corollary}
  Unless~$\NP \subseteq \coNPPoly$, \DMSE{} on planar acyclic digraphs with parameter $\kappa(G,s,t,p,k) \coloneqq p+k+\Delta_{\rm in}(G)+\Delta_{\rm out}(G)$ does not admit a polynomial kernel.
\end{corollary}

\section{Conclusion}
\label{sec:conclusion}
On the positive side, we proved that \MSElong{} on bounded grids is solvable in linear time when both dimensions are either small or large compared to the number~$p$ of paths.
On the negative side, we proved that \MSE{} becomes $\NP$-hard on subgraphs of the bounded grid, even if the subgraph is directed and acyclic, and that it does not allow for polynomial kernels on planar graphs when parametrised by a  combined parameter~$k+p+\Delta+\diamG{}+\tw$, unless $\NP \subseteq \coNP/\text{poly}$.

We conjecture that \MSE{} on $p$-narrow grids is solvable in polynomial time.
In particular, we find it interesting whether an arithmetic criterion similar to the $p$-large case~(cf.~Lemma~\ref{lem:constr_large_k}) exists.
Furthermore, in our reduction from \Probl{Vertex Cover}, the construction yields a grid with a large amount of edges removed by taking a subgraph.
Is \MSE{} parametrised by the number of edges removed from the grid in FPT (or even admits a polynomial-size problem kernel)?

We consider it as interesting to study \DMSE{} on \emph{Manhattan street networks} (cf.~\cite{maxemchuk1987routing}).
Recently, \MSE{} is considered with an additional time-aspect~\cite{Mor16}.
Herein, on a high level, an edge is shared if it appears in at least two paths at the same time.
Another future research direction could be to study \MSE{} with the additional time-aspect on grid-like graphs.

\bibliographystyle{plain}
\bibliography{planarMSE-arxiv-V2}

\appendix

\end{document}